\documentclass{svproc}
\usepackage{url}
\usepackage{amsmath,amssymb,amsfonts}

\begin{document}
\mainmatter              
\title{MathPartner: an artificial intelligence cloud service}
\titlerunning{ MathPartner: an artificial intelligence cloud service}  
%
\author{Gennadi Malaschonok\inst{1} \and Alexandr Seliverstov\inst{2}}
\authorrunning{ Gennadi Malaschonok, Alexandr Seliverstov} 
%
\tocauthor{Gennadi Malaschonok, Alexandr Seliverstov}
\institute{ National University of Kyiv-Mohyla Academy, 04070  Kyiv, Ukraine\\
\email{malaschonok@ukma.edu.ua},\\ WWW home page:
\texttt{https://www.ukma.edu.ua/eng/ }
\and
Institute for Information Transmission Problems,
Russian Academy of Sciences (Kharkevich Institute), \\
Bolshoi Karetny per. 19, 127994 Moscow, Russia }


\maketitle              

\begin{abstract}
 In a broad sense, artificial intelligence is a service to find a solution to complex intellectual problems. In this sense, the MathPartner service provides artificial intelligence that allows us to formulate questions and receive answers to questions formulated in a mathematical language. For mathematicians and physicists today, such a language is \LaTeX. The MathPartner service uses a dialect of \LaTeX, which is called Mathpar. The service is a cloud-based computer algebra system and provides users with the opportunity to solve many mathematical problems. In this publication, we focus only on a small class of extremum problems, which are widely applied in economics, management, logistics, and in many engineering fields. In particular, we consider the shortest path problem and discuss an algorithm that is based on the tropical mathematics. The ability to work with many types of classical and tropical algebras, which are freely available to users, is an important distinguishing feature of this intelligent tool for symbolic-numerical calculations. We also consider the use of the simplex algorithm for solving optimization problems.
\end{abstract}
 
\section{Introduction} 
 
 
Intelligent systems and artificial intelligence should at least have two properties: reflect the real world and interact with a person. 
If we talk about creating a universal intellectual system in the future, then it must have natural science knowledge, and primarily mathematical knowledge.

While humanity is still far from creating an artificial intelligence that would have mathematical knowledge, but such a task should already be formulated as one of the primary tasks. What is the traditional way of using mathematical knowledge? First, the problem is formulated in a mathematical language, and then the solution of the problem is found using known methods and algorithms. Therefore, the first step should be the creation of an intelligent system that is capable of solving traditional mathematical problems, if they are correctly formulated in traditional mathematical language. 

Here are some simple examples. A weighted graph is given, two of its vertices are indicated, and it is required to find the shortest distance between these two vertices. A system of linear algebraic equations is given, a linear form is indicated, and it is required to find a solution to the system that maximizes (or minimizes) this linear form. A system of algebraic equations is given and it is required to find all its solutions in a specific numerical region. A system of differential equations is given and it is required to find all its solutions in a certain class of functions.

Of course, a program that can solve such problems is not artificial intelligence. But such a program is the main component without which it is impossible to create artificial intelligence with mathematical knowledge. Today, the most convenient form is the cloud form. The creation of such a cloud component is a very important step towards the creation of artificial intelligence.

We propose to use MathPartner cloud mathematics for this purpose. It is freely available~\cite{MP}. To get acquainted with its capabilities, you can read the user manual as well as journal articles~\cite{M17,MS21,MS21A}. However, the possibilities of solving extreme problems were very little covered.

Therefore, the purpose of this article is to acquaint the reader with the capabilities that this service provides for solving extreme problems.

To solve extremal problems, you can use both classical and tropical algebras. Note that such algebras are freely available to users and this is one of the important advantages of this service.

Calculations over max-plus as well as min-plus semirings play an important role both in discrete optimization and for the study of systems of equations over the field of complex numbers. These semirings are known as tropical semirings and has been intensively studied in recent decades. In particular, they forms the basis of tropical geometry~\cite{LM20,G20,GP20,VVY21,G22,GR22,NYW22,YC22}. A new branch of linear algebra was also developed~\cite{LRSS19,FG20,DO20,WFSS20,J21,X21,YY22}.

The shortest path problem has many applications in artificial intelligence and operations research. Classical algorithms such as the Dijkstra algorithm are known. This problem is effectively solvable over the min-plus semiring. 

For both linear functional and polyhedron defined by linear functions with integer coefficients, the optimization problem can be reformulated over max-plus or min-plus semirings. Therefore, such problems about polyhedra provide a large set of problems related to tropical mathematics. Although in the MathPartner system, the simplex algorithm is implemented over ordinary fields. The average-case complexity of the algorithm is polynomial, but it is high in the worst case~\cite{S83,VS83}. The simplex algorithm is widely used, although other methods are more efficient in the worst case~\cite{SFK18,CLS20,C21}. Note that the simplex algorithm is sensitive to rounding off at intermediate computation steps~\cite{O88}. For solving feasibility in linear programs with up to two variables per inequality, $n$ variables, and $m$ inequalities, there is an $O(mn^2 \log m)$ algorithm which is derived from the Fourier--Motzkin elimination method~\cite{HN94,SW21}. 


This paper is structured as follows. Section~2 provides basic definitions and theoretical background.
Section~3 describes implementation of some tropical algorithms.
 Section~4 describes implementation of the simplex algorithm. Section~5 concludes the paper.
 
\section{ Basic definitions and theoretical background}
\subsection{The max-plus semiring}

As usual, tropical algebra refers to the max-plus semiring, i.e., the set 
${\Bbb R} \cup\{-\infty\}$ 
equipped with the addition $x\oplus y=\max(x, y)$ as well as the multiplication $x \odot y=x + y$. For all $x$, both identities $\max(-\infty,x)=x$ and $-\infty+x=-\infty$ hold. Every $x\ne-\infty$ has the multiplicative inverse $-x$. So, the algebra is known as a semifield~\cite{LRSS19}.

The matrix multiplication over the max-plus semiring is defined as usual. The multiplication symbol between matrices is omitted. For example,
$$\left(
\begin{array}{cc}
A_{11} & A_{12}\\
A_{21} & A_{22} \\
\end{array}\right)
\left(
\begin{array}{c}
x_{1}\\
x_{2}\\
\end{array}\right)=
\left(
\begin{array}{c}
\max(A_{11}+x_1, A_{12}+x_2) \\
\max(A_{21}+x_1, A_{22}+x_2) \\
\end{array}\right).$$

The $\max()$ operation acts on matrices of the same size entry-wise, i.e., $$\max(A,B)_{jk}=\max(A_{jk},B_{jk}).$$ 

\begin{definition}
Given a $n\times n$ matrix $A$ over the max-plus semiring. Its pseudo-inverse matrix $A^{-}$ is a $n\times n$ matrix with entries
$$A^{-}_{jk}=\left\{
\begin{array}{rl}
-A_{kj}, & A_{kj}\ne -\infty\\
-\infty, & A_{kj}= -\infty\\
\end{array}\right.$$
\end{definition}

Of course, $(A^{-})^{-}=A$. The pseudo-inverse matrix is defined for every square matrix. It can be calculated easier than the inverse matrix over a field. Nevertheless, it allows to solve tropical equations.

Let $\mathrm{diag}(d_1,\dots,d_n)$ denote the matrix with entries
$$D_{jk}=\left\{
\begin{array}{rl}
d_k, & j=k\\
-\infty, & j\ne k\\
\end{array}
\right.$$
and $I=\mathrm{diag}(0,\dots,0)$.

\begin{remark}
If $A=\mathrm{diag}(d_1,\dots,d_n)$, where all $d_k\ne -\infty$, then $AA^{-}=I$, i.e., $A^{-}$ is the inverse matrix.
\end{remark}

\begin{proposition}
For a $n\times n$ matrix $A$ and a column $b$, every column $x\le (b^{-}A)^{-}$ is a particular solution to the inequality $Ax\le b$. Moreover, the upper bound on the computation complexity equals $O(n^2)$ operations over the max-plus semiring. 
\end{proposition}
\begin{proof}
If $x \le (b^{-}A)^{-}$, then for every index $k$ the inequality $x_k\le (\max_j(b_j^{-}+A_{jk}))^{-}$ holds. So, $x_k\le \min_j(b_j+A_{jk}^{-})$.
Next, for all $k$ and $j$, $x_k\le b_j+A_{jk}^{-}$ and $A_{jk}+x_k\le b_j$.
So, for every index $j$, $\max_k(A_{jk}+x_k)\le b_j$. Thus, $Ax\le b$.
\end{proof}

\begin{proposition}
For a $n\times n$ matrix $A$ and a column $b$, if there exists a solution to the equation $Ax=b$, then a particular solution can be calculated as $x=(b^{-}A)^{-}$. Moreover, the upper bound on the computation complexity equals $O(n^2)$ operations over the max-plus semiring. 
\end{proposition}

\begin{remark}
For a sparse matrix $A$, a particular solution to the equation $Ax=b$ can be calculated quicker.
\end{remark}

There are many types of equations because the additive operation is not invertible. An equation of the type $Ax=x$ is  known as the homogeneous Bellman equation. For a column $b$, an inhomogeneous equation is one of the type $\max(Ax,b)=x$.

\begin{definition}
A partial operation $()^{\times}$ is called the closure when both identities hold:
\begin{enumerate}
\item
If $\max(X,Y)=Y$ and both $X^{\times}$ and $Y^{\times}$ exist, then $\max(X^{\times},Y^{\times})=Y^{\times}$.
\item
$X^{\times}=\max(I,XX^{\times})=\max(I,X^{\times}X)$.
\end{enumerate}
\end{definition}

\begin{proposition}
Given both $n \times n$ matrices $A$ and $B$ over max-plus semiring. If the closure $A^{\times}$ exists, then the matrix $X=A^{\times}B$ is a solution to the matrix Bellman equation $\max(AX,B)=X$.
\end{proposition}
\begin{proof}
$\max(AX,B)=\max(AA^{\times}B,B)=\max(AA^{\times},I)B=A^{\times}B=X$.
\end{proof}

The closure can be calculated as $A^{\times}=\max(I,A,A^2,\dots)$. For elements of the max-plus semiring, the closure is
$$x^{\times}=\left\{\begin{array}{ll}
0, & x\le 0\\
\mathrm{undefined}, & x>0\\ 
\end{array}\right.$$
Moreover, for $n \times n$ matrices, if the closure exists, then it is equal to finite expression $A^{\times}=\max(I,A,\dots,A^{n-1})$.

The next proposition provides a block-recursive algorithm to calculate the closure of a matrix.

\begin{proposition}  
Given a block matrix over the max-plus semiring
$$A =\left(
\begin{array}{cc}
E & F\\
G & H\\
\end{array}
\right)$$
If both $E^{\times}$ and $D^{\times}$ exist, where $D=\max(H,GE^{\times}F)$, then the closure $A^{\times}$ can be found using closures of the blocks:
$$A^{\times} =\left(
\begin{array}{cc}
E^{\times}\max(I,FD^{\times}GE^{\times}) & E^{\times}FD^{\times}\\
D^{\times}GE^{\times} & D^{\times}\\
\end{array}
\right).$$
\end{proposition}
\begin{proof} 
Let $-\infty$ denotes the matrix whose entries are equal to $-\infty$.
First of all, it is easy to check that the block expression for $A^{\times}$ has the following factorization: 
$$\left(
\begin{array}{cc}
E^{\times}\max(I, FD^{\times}GE^{\times}) & E^{\times}FD^{\times}\\
D^{\times}GE^{\times} & D^{\times}\\
\end{array}
\right)=
$$
$$
\left(
\begin{array}{cc}
I & E^{\times}F\\
-\infty & I\\
\end{array}
\right)
\left(
\begin{array}{cc}
I & -\infty\\
-\infty & D^{\times}\\
\end{array}
\right)
\left(
\begin{array}{cc}
I & -\infty\\
G & I\\
\end{array}
\right)
\left(
\begin{array}{cc}
E^{\times} & -\infty\\
-\infty & I\\
\end{array}
\right).
$$

Let us prove the identity $A^{\times}=\max(I,A^{\times}A)$.
Consider the product of the first two matrices and matrix $A$:
$$
\left(
\begin{array}{cc}
I & -\infty\\
G & I\\
\end{array}
\right)
\left(
\begin{array}{cc}
E^{\times} & -\infty\\
-\infty & I\\
\end{array}
\right)
\left(
\begin{array}{cc}
E & F\\
G & H\\
\end{array}
\right)
=
\left(
\begin{array}{cc}
E^{\times}E & E^{\times}F\\
GE^{\times} & D\\
\end{array}
\right).
$$
It follows from here that
$$
A^{\times}A
=
\left(
\begin{array}{cc}
I & E^{\times}F\\
-\infty & I\\
\end{array}
\right)
\left(
\begin{array}{cc}
I & -\infty\\
-\infty & D^{\times}\\
\end{array}
\right)
\left(
\begin{array}{cc}
E^{\times}E & E^{\times}F\\
GE^{\times} & D\\
\end{array}
\right)
$$
 $$
 =\left(
\begin{array}{cc}
\max(E^{\times}E, E^{\times}FD^{\times}GE^{\times}) & E^{\times}FD^{\times}\\
D^{\times}GE^{\times} & D^{\times}D\\
\end{array}
\right).
$$
Therefor 
$$
 \max(I,A^{\times} A)
 =
\max\bigg(
\left(\begin{array}{cc} I
& -\infty\\
-\infty & I
\end{array}
\right), 
\left(
 \begin{array}{cc}
\max(E^{\times}E, E^{\times}FD^{\times}GE^{\times}) & E^{\times}FD^{\times}\\
D^{\times}GE^{\times} & D^{\times}D\\
\end{array}
\right)
\bigg)
$$
$$=
 \left(
\begin{array}{cc}
E^{\times}\max(I, FD^{\times}GE^{\times}) & E^{\times}FD^{\times}\\
D^{\times}GE^{\times} & D^{\times}\\
\end{array}\right).
$$
The identity $A^{\times}=\max(I,AA^{\times})$ can also be proved similarly.
\end{proof}

\subsection{The min-plus semiring}
In the same way one can consider the min-plus semiring 
${\Bbb R}\cup\{+\infty\}$ with addition $x\oplus y=\min(x,y)$ as well as multiplication $x \odot y=x + y$. For all $x$, both identities $\min(+\infty,x)=x$ and $+\infty+x=+\infty$ hold. Every $x\ne+\infty$ has the multiplicative inverse $-x$. 
Of course, in all definitions $\max()$ is replaced by $\min()$ as well as $-\infty$ is replaced by $+\infty$. For example, the pseudo-inverse matrix $A^{-}$ over the min-plus semiring is a matrix with entries
$$A^{-}_{jk}=\left\{
\begin{array}{rl}
-A_{kj}, & A_{kj}\ne +\infty\\
+\infty, & A_{kj}= +\infty\\
\end{array}\right.$$
The diagonal matrix $\mathrm{diag}(d_1,\dots,d_n)$ has entries
$$D_{jk}=\left\{
\begin{array}{rl}
d_k, & j=k\\
+\infty, & j\ne k\\
\end{array}
\right.$$
The closure is defined as $A^{\times}=\min(I,A,A^2,\dots)$.

Note that the order over the min-plus semiring is reversed.

Both max-plus and min-plus semirings are tangled by the sister identities $(\max(A,B))^{-}=\min(A^{-},B^{-})$ and $(\min(A,B))^{-}=\max(A^{-},B^{-})$. 

The shortest path problem is the problem of finding a path between two vertices in a graph such that the sum of the weights of its constituent edges is minimized. Let the weights be entries of the matrix $A$. If all  weights are nonnegative, then the problem can be reduced to calculating the closure $A^{\times}$ over the min-plus semiring. In fact, it uses only finite sequence of matrices in the equation $A^{\times}=\min(I,A,A^2,\dots)$ because the length of the shortest path is bounded. Of course, if all  weights are nonpositive, then one can find the maximum using  the max-plus semiring, but this formulation of the problem is not usually used. For nonnegative weights, the travelling salesman problem can also be expressed over the max-plus semiring, but its computational complexity seems to be very high because closely related the Hamiltonian path problem is NP-complete.

\subsection{The algorithm for calculating the matrix closure}  

We formulate the algorithm for calculating the matrix closure. 
Let us designate:
\begin{itemize}
\item
$M$ is an input matrix whose size $n$ is a power of $2$;
\item
$\mathtt{ring}$ is an algebra type (one of the tropical semirings or classical fields);
\item
$\mathtt{Split()}$ is procedure for dividing a matrix into four equal blocks;
\item
$\mathtt{Join()}$ is a procedure for joining four blocks into one matrix;
\item
$\mathtt{Inverse()}$ is the matrix inversion procedure over a field;
\item
$I$ is the identity matrix;
\item
both $+$ and $*$ are basic operations over a (semi)ring.
\end{itemize}
 
\begin{verbatim}
 closure(M, ring) {
   if (ring is tropical) 
   then
    if(M.size=1) AND (M=[[a]]) 
    then{b=closure(a);return [[b]]}
    else{   
         (E,F,G,H) = Split(M):  
         S = closure(E,ring);
         B = G*S;
         R4 = closure(H+B*F, ring);
         R3 = R4+B;
         V = S*F ;
         R2 = V*R4;
         R1 = S+V*R3;
         return Join(R1, R2, R3, R4);
    }
   else return Inverse(I-M); 
 }
\end{verbatim}

In contrast to matrices over a field, there is not known fast algorithm for matrix multiplication over tropical semirings similar to the Strassen algorithm~\cite{S69,KS20} because the additive operation is not invertible.
 
Let $\gamma$ and $\beta$ be constants, $3\geq\beta>2$, and let
${\mathcal M}(n)= \gamma n^{\beta} + o(n^{\beta})$ be the number of  arithmetic operations in one $n\times n$ matrix multiplication. We know today algorithm with $\gamma=1$ and  ${\beta}=3$.
 
\begin{proposition}
The  closure algorithm has the complexity of matrix multiplication.
\end{proposition}
\begin{proof}
For calculating the closure of the $n\times n$ matrix we have to compute five matrix multiplications and two matrix closure for the $(n/2) \times (n/2)$ blocks.    

So, we get the following recurrent equality for complexity
$$
t(n)=2 t(n/2)+ 5 {\mathcal M}(n/2), t(1)=1.
$$

After summation from $n=2^k$ to $2^1$ we obtain
$$
5 \gamma\left(2^0 2^{\beta(k-1)} + \cdots + 2^{k-1}2^{\beta\cdot 0}+2^{k}\right) t(1) 
$$
$$
= n+\frac{5\gamma n}{2} \sum_{i=0}^{k-1} 2^{i(\beta-1)} = 5 \gamma\frac{n^{\beta}-n}{2^{\beta}-2} + n.
$$
Therefore the complexity of the closure algorithm is 
$
\sim
\frac{5 \gamma n^{\beta}}{2^{\beta}-2}$. For $\gamma=1$ and ${\beta}=3$, we get $\sim \frac 5 6 n^3$.
\end{proof}

\section{Tropical mathematics in MathPartner}
\subsection{Algebras in MathPartner}

Working in the MathPartner computer algebra system involves choosing the algebra over which the calculations are performed. The default algebra is the algebra $R64[x,y,z,t]$ of polynomials in four variables with the monomial order, where inequalities $x<y<z<t$ hold. $R64$ is the set of floating-point 64-bit numbers having $52$ bits for mantissa, $11$ bits for the exponent and one bit for the sign. Instead of floating-point numbers $R64$, one can use integers $Z$, rational numbers $Q$, floating-point numbers of arbitrary precision $R$, complex floating-point numbers $C64$ and $C$, modulo prime residues $Zp$ and $Zp32$, as well as idempotent algebras, including $ZMaxPlus$, $ZMinPlus$, $RMaxPlus$, $RMinPlus$, $R64MaxPlus$, and $R64MinPlus$. The algebra can be chosen with a command like $\mathtt{SPACE=ZMaxPlus[]}$. Many commands only work in the proper algebra.

The additive operation is always denoted by $+$, and the multiplicative operation is denoted by $*$.
 
\begin{example}
Let us run $\mathtt{SPACE = ZMaxPlus[]; 2+3;}$ the output is $3=\max(2,3)$. 
On the other hand, let us run $\mathtt{SPACE = ZMaxPlus[]; 2*3;}$ the output is $5$.
\end{example}

\subsection{Simple equations and inequalities}

For given both matrix $A$ and vector $b$, to find a particular solution to the system of linear algebraic equations $Ax = b$ over a tropical semiring one can run $\backslash\mathtt{solveLAETropic}(A, b)$.

\begin{example} 
Let us consider the matrix
$$A=\left(\begin{array}{cc}1 & 2\\ 3 & 0\end{array}\right)$$
and the column
$$b=\left(\begin{array}{c}5\\7\end{array}\right).$$
Let us run
$$\mathtt{SPACE = ZMaxPlus[]; A = [[1, 2],[3, 0]]; b = [5, 7];}
$$
$$
 \mathtt{ \backslash{solveLAETropic}(A, b);}$$
The output is equal to 
$$\left(\begin{array}{c}4\\3\end{array}\right).$$
\end{example}

In the same way, to solve the system of linear algebraic inequalities one can run $\backslash\mathtt{solveLAITropic}(A, b)$.

\begin{example}
Let us run
$$\mathtt{SPACE = ZMaxPlus[]; A = [[2,0],[3,1]]; b = [1,1];}
$$
$$\mathtt{ \backslash{solveLAITropic}(A, b);}$$
The output is equal to $[[-\infty,-2], [-\infty,0]]$.
\end{example}

\subsection{The closure}

The closure $A^{\times}$ can be calculated by the $\backslash\mathtt{closure}()$ operation. For numbers, if the closure does not exists, then the output is $\infty$.

\begin{example}
Let us run $\mathtt{SPACE = ZMaxPlus[]; \backslash{closure}(1);}$ the output is $\infty$. 
On the other hand, let us run $\mathtt{SPACE = ZMaxPlus[]; \backslash{closure}(-1);}$ the output is $0$.
\end{example}

\begin{example}
Let us calculate the closure of the matrix 
$$A=\left(\begin{array}{cc}
0 & 1 \\
2 & 0 \\
\end{array}\right)$$
over the min-plus semiring.
$$\mathtt{SPACE = ZMinPlus[]; A=[[0,1],[2,0]]; \backslash{closure}(A);}$$
The output is equal to
$$A^{\times}=\left(\begin{array}{cc}
0 & 1 \\
2 & 0 \\
\end{array}\right)$$
\end{example}

\begin{example}
Let us calculate the closure of the matrix 
$$A=\left(\begin{array}{cc}
-1 & -2 \\
-3 & -4 \\
\end{array}\right)$$
over the max-plus semiring.
$$\mathtt{SPACE = ZMaxPlus[]; A=[[-1,-2],[-3,-4]]; \backslash{closure}(A);}$$
The output is equal to
$$A^{\times}=\left(\begin{array}{cc}
0 & -2 \\
-3 & 0 \\
\end{array}\right)$$
\end{example}

\subsection{The Bellman equations and inequalities}

To find a solution to a homogeneous Bellman equation $Ax=x$ one can run $\backslash\mathtt{BellmanEquation}(A)$. For a column $b$, to find a solution to an inhomogeneous Bellman equation one can run $\backslash\mathtt{BellmanEquation}(A,b)$. 

Two commands $\backslash\mathtt{BellmanInequality}(A)$ and $\backslash\mathtt{BellmanInequality}(A,b)$ allows to find a solution to Bellman inequalities $Ax\le x$ and $\max(Ax,b)\le x$, respectively.

\subsection{The shortest path}

Let $A$ be a matrix with nonnegative or infinite entries, where all entries at the leading diagonal vanish, i.e., for all $k$, $x_{kk}=0$. Its entry can be interpreted as the graph edge weight or length. However, the triangle inequality can be violated. If there is no edge between two vertices $k$ and $j$, then $x_{kj}=\infty$. 
To find the shortest path from $k$ to $j$ one can run $\backslash\mathtt{findTheShortestPath}(A, k, j)$ over the min-plus algebra. To calculate the smallest distances for all vertex pairs, one can run $\backslash\mathtt{searchLeastDistances}(A)$.

\begin{example}
Let us consider the graph with the adjacency matrix  
$$A=\left(\begin{array}{ccc}
0 & 1 & \infty \\
1 & 0 & 1 \\
\infty & 1 & 0 \\
\end{array}\right),$$
i.e., the graph is a path. Let us run 
$$\mathtt{SPACE = ZMinPlus[]; A=[[0,1,\infty],[1,0,1],[\infty,1,0]];}
$$
$$\mathtt{ \backslash{searchLeastDistances}(A);}$$
The output is equal to 
$$\left(\begin{array}{ccc}
0 & 1 & 2 \\
1 & 0 & 1 \\
2 & 1 & 0 \\
\end{array}\right)$$
Next, let us run $\backslash\mathtt{findTheShortestPath}(A, 2, 1)$. The output is equal to $[2,1]$, that is, the shortest path from $2$ to $1$ consist of two vertices $2$ and $1$.
\end{example}

\section{Linear programming in MathPartner}
\subsection{Systems of algebraic inequalities}

To calculate the solution to a system of univariate algebraic inequalities one can run $\backslash\mathtt{solve}()$ over a ring of polynomials. 

\begin{example} 
To solve the system of two inequalities $x-6 > 0$ and $x-7 < 0$, let us run 
$$\mathtt{SPACE = Q[x]; b = \backslash {solve}([x-6 > 0, x-7 < 0]);}$$
The output is equal to the interval $(6,7)$. Note that parentheses denote open interval boundaries, and square brackets denote closed interval boundaries.
\end{example}

\subsection{The simplex algorithm}

To solve a linear programming problem, one can execute either the $\backslash\mathtt{SimplexMax}()$ command or its counterpart that is the $\backslash\mathtt{SimplexMin}()$ command. The result is a vector $x$. Calculations should be carried out over either rational or floating-point numbers.

If the objective function $c^Tx$ is to be maximized ($c^Tx \rightarrow \max$) under the conditions $Ax \le b$ and $x\ge 0$, then one can run $\backslash\mathtt{SimplexMax}{(A, b, c)}$.

\begin{example}
Let us maximize the $3x_1+x_2+2x_3$ under the conditions $x_1\ge 0$, $x_2\ge 0$, $x_3\ge 0$, $x_1+x_2+3x_3\le 30$, $2x_1+2x_2+5x_3\le 24$, and $4x_1+x_2+2x_3\le 36$. Let us run
$$\begin{array}{l}
SPACE = R64[];\\
A = [[1, 1, 3], [2, 2, 5], [4, 1, 2]];\\
b = [30, 24, 36];\\
c = [ 3,  1,  2];\\
x = \backslash\mathtt{SimplexMax}(A, b, c);\\
\end{array}$$
The output is $[8, 4, 0]^T.$
\end{example}

Of course, only symbolic calculations over rational numbers without rounding guarantee the answer according to well-known results~\cite{O88}.

If the objective function $c^Tx$ is to be maximized ($c^Tx \rightarrow \max$) under the conditions $A_1x \le b_1$, $A_2 x =  b_2$, and $x \ge 0$, then one can run 
$$\backslash\mathtt{SimplexMax}{(A_1, A_2, b_1, b_2, c)}.$$

If the objective function $c^Tx$ is to be maximized ($c^Tx \rightarrow \max$) under the conditions $A_1x \le b_1$, $A_2 x =  b_2$, $A_3 x \ge  b_3$, and $x \ge 0$, then one can run 
$$\backslash\mathtt{SimplexMax}{(A_1, A_2, A_3, b_1, b_2, b_3, c)}.$$ Some matrices can be empty $()$. If the objective function $c^Tx$ is to be minimized ($c^Tx \rightarrow \min$), then one can run either $\backslash\mathtt{SimplexMin}{(A, b, c)}$ or 
$$\backslash\mathtt{SimplexMin}{(A_1,A_2, b_1, b_2, c)}$$ 
or 
$$\backslash\mathtt{SimplexMin}{(A_1, A_2, A_3, b_1, b_2, b_3, c)},$$ respectively.

\section{Discussion}

We do not give a detailed description of the MathPartner, referring the reader to other sources and manuals, but we dwell in detail on the issue of solving extreme problems.  

But still, we note that MathPartner can solve a large class of classical symbolic and numerical problems in mathematical analysis, algebra, graph theory, the theory of differential equations and much more, including the construction of graphics, the depiction of three-dimensional figures, the animation of images, etc. 

For example, MathPartner can solve many matrix problems for number fields and polynomial rings. It allows you to factor matrices, find the Bruhat decomposition, LSU, QR, SVD, Cholesky decomposition. The list of matrix functions includes calculation of the inverse and adjoint matrix, calculation of the kernel of a linear operator, Moore-Penrose inverse matrix and pseudo-inverse matrix, calculation of the determinant and echelon form of the matrix, calculation of the characteristic polynomial and others.

For detailed information, we recommend user manual http://mathpar.com/ downloads/MathparHandbook\_en.pdf.

\section{Conclusion}

Artificial intelligence, in the form in which it is already available for experimentation today, makes it possible to solve a large class of problems related to text and image processing. But it remains completely helpless when faced with even the simplest problems in the field of mathematics or physics.

We formulate the problem of creating artificial intelligence that can solve mathematical problems and propose to consider the MathPartner service as a possible basis for its creation.

In fact, we can talk about AI, which will have three global levels. 

The first level is text processing and isolating those fragments that formulate a mathematical problem. 

The second level is the reformulation of the problem into language Mathpar. 

And the third level is the solution of the problem using Cloud Mathematics MathPartner. 

When the solution is obtained at the third level, its text and appearance are already familiar to specialists. It will be enough to return it to the first level and add the necessary text environment.

\section{Acknowledgment}

The first author's work was supported by the grants from the Simons Foundation. Awards ID:
1030285 and 1290592. My deepest sympathies go to the family of Jim Simons, a
great man and brilliant mathematician, founder of this foundation, who passed
away this year on May 10.

 
%


%
%

\end{document}